\documentclass[11pt,fullpage]{article}

\usepackage{amsfonts}
\usepackage{amsmath}
\usepackage{amssymb}
\usepackage[dvips]{graphicx}
\usepackage{epsf}
\usepackage{epsfig}
\usepackage{subfigure}
\usepackage{wrapfig}
\usepackage{geometry}
\usepackage{algorithm}
\usepackage{algorithmic}
\usepackage{color}
\usepackage{tweaklist}

\renewcommand{\itemhook}{\setlength{\topsep}{4pt}%
  \setlength{\itemsep}{4pt} }

\renewcommand{\enumhook}{\setlength{\topsep}{4pt}%
  \setlength{\itemsep}{4pt} }

\newtheorem{theorem}{Theorem}[section]
\newtheorem{lemma}[theorem]{Lemma}
\newtheorem{proposition}[theorem]{Proposition}
\newtheorem{corollary}[theorem]{Corollary}

\newtheorem{fact}[theorem]{Fact}

\newenvironment{proof}[1][Proof:]{\begin{trivlist}
\item[\hskip \labelsep {\bfseries #1}]}{\qed \end{trivlist}}

\def\squarebox#1{\hbox to #1{\hfill\vbox to #1{\vfill}}}
\newcommand{\qed}{\hspace*{\fill}
\vbox{\hrule\hbox{\vrule\squarebox{.667em}\vrule}\hrule}\smallskip}

\newcommand{\N}{\ensuremath{\mathbb{N}}} 
\newcommand{\Z}{\ensuremath{\mathbb{Z}}} 
\newcommand{\R}{\ensuremath{\mathbb{R}}} 

\newcommand{\bigo}{{\mathcal{O}}}
\newcommand{\cd}{{\mathcal{C}}}
\newcommand{\ct}{{v}}
\newcommand{\profit}{{\pi}}
\newcommand{\expect}{{\mathbb{E}}}
\newcommand{\prob}{{\mbox{Prob}}}

\DeclareMathOperator{\argmax}{argmax}

\def \ee {{\rm e}} 

\begin{document}

\title{Pricing Randomized Allocations}
\author{
Patrick Briest \footnotemark[1]  \and
Shuchi Chawla \footnotemark[2]  \and
Robert Kleinberg \footnotemark[1] \and
S. Matthew Weinberg \footnotemark[1] \and
}
\footnotetext[1]{Department of Computer Science, Cornell University, Ithaca NY 14853.  E-mail: \texttt{\{pbriest,rdk\}@cs.cornell.edu, smw79@cornell.edu}. The first author is supported by a scholarship of the German Academic Exchange Service (DAAD).}
\footnotetext[2]{Computer Science Department, University of Wisconsin - Madison, 1210 W. Dayton Street, Madison, WI 53706. E-mail: \texttt{shuchi@cs.wisc.edu}.}

\date{}

\maketitle

\begin{abstract}
Randomized mechanisms, which map a set of bids to a 
probability distribution over outcomes rather than a
single outcome, are an important but ill-understood
area of computational mechanism design.  
We investigate the role of randomized outcomes 
(henceforth, ``lotteries'') in
the context of a fundamental and archetypical multi-parameter
mechanism design problem: selling heterogeneous items to 
unit-demand bidders.  
To what extent can a seller 
improve its revenue by pricing lotteries rather than 
items, and does this modification of the problem affect
its computational tractability?  Our results show that
the answers to these questions hinge on whether 
consumers can purchase only one lottery (the
buy-one model) or purchase any set of lotteries 
and receive an independent sample from each (the buy-many model).
In the buy-one model, there is a polynomial-time
algorithm to compute the revenue-maximizing envy-free 
prices (thus overcoming the inapproximability of the
corresponding item pricing problem) and the revenue
of the optimal lottery system can exceed the revenue
of the optimal item pricing by an unbounded factor
as long as the number of item types exceeds 4.
In the buy-many model with $n$ item types, 
the profit achieved by lottery pricing can exceed
item pricing by a factor of $\Theta(\log n)$ but
not more, and optimal lottery pricing cannot be approximated 
within a factor of $\bigo (n^{\varepsilon})$ for some $\varepsilon >0$, unless NP $\subseteq \bigcap _{\delta >0}$ BPTIME$(2^{\bigo (n^{\delta})})$.  Our lower bounds rely on 
a mixture of geometric and algebraic techniques,
whereas the upper bounds use a novel rounding scheme to 
transform a mechanism with randomized outcomes into one 
with deterministic outcomes while losing only a bounded
amount of revenue.
\end{abstract}

\newpage
\setcounter{page}{1}

\section{Introduction}
\label{sec:intro}

It is well known that randomness is a very useful resource 
for designing efficient algorithms, and the same is true
for designing truthful mechanisms to allocate resources among self-interested participants 
with private inputs.
Randomness plays several important roles in mechanism design.
First, while approximately optimal 
deterministic truthful mechanisms are required to solve
intractable problems in some cases,
computationally efficient randomized truthful mechanisms 
can often provide much better approximations. 
Second, in certain prior-free or online settings,
randomness can offset the cost imposed by a lack of
knowledge of agents' types (private inputs), allowing
mechanisms to compete favorably against an 
optimal omniscient mechanism when deterministic 
mechanisms cannot.  In the absence of these two effects, 
i.e.~when computational efficiency is not a consideration 
and the designer has distributional information about the 
agents' types (the Bayesian setting), one might suspect 
that randomness provides no benefit.  Riley
and Zeckhauser~\cite{RZ} proved that this is indeed true for revenue
maximization in single-parameter settings. Surprisingly this
observation fails to extend to multi-parameter settings:
Thanassoulis~\cite{Than} and Manelli and Vincent~\cite{MV}
independently noted that in multi-parameter settings a mechanism can
use randomized allocations to extract greater revenue from the
agents than any deterministic allocation. 
The benefit arises by allowing the mechanism to 
price-discriminate among agents to a greater extent by expanding its
allocation space to include probabilistic mixtures of outcomes. This
use of randomization has not yet been studied from a computer science
perspective.  Here we study randomized allocations (henceforth,
``lotteries'') in the context of a fundamental and archetypical
multi-parameter mechanism design problem: selling heterogeneous items
to unit-demand bidders.


There are a few reasons why we focus on unit-demand bidders and
heterogeneous items.  First, it is the setting explored by
Thanassoulis in his pioneering work on lotteries as an optimal selling
mechanism.  That work raised fundamental questions about the extent to
which lotteries can improve a seller's revenue; here we answer these
questions.  Second, unit-demand multi-product pricing is one of the
best known inapproximable problems in the theory of pricing
algorithms~\cite{Briest,envyfree}, and is thus
a test case for the general question, ``Does pricing
lotteries instead of items allow us to circumvent the computational
hardness of pricing problems?''  We will see that the answer is
affirmative.

In more detail, the problems we consider in this paper involve selling
$n$ types of items.  A consumer is represented by a vector of $n$
non-negative real numbers denoting her valuation for receiving one
copy of each item type.  A lottery is specified by a price and a
probability distribution over the set of item types.  Consumers are
risk neutral and quasi-linear: their utility for a lottery is equal to
their expected valuation for a random sample from the item
distribution, minus the price of the lottery.  They have unit demand:
their valuation for a set of items is equal to their maximum valuation
for any element of that set.  Given a distribution over
consumers\footnote{With the exception of Section~\ref{polytimeAlgo},
  we allow distributions with infinite support.} and a set of
lotteries (a \emph{lottery pricing system}) the seller achieves a
revenue equal to the expected price paid by a consumer who is randomly
sampled from the given distribution and chooses her
utility-maximizing\footnote{If the set of lotteries is infinite, we
  require it to be a closed subset of the space of all lotteries,
  topologized as a subset of $\R^{n+1}$.  This requirement is
  necessary in order to ensure that there exists a utility-maximizing
  lottery.} lottery.

The following example from~\cite{Than}
illustrates the power of lottery pricing systems.
Suppose there are two item types and a consumer's
valuations for the two items are independent and
uniformly distributed in an interval $[a,b]$.  
The optimal item pricing always sets the same
price $p^*$ for both items; the revenue-maximizing 
value of $p^*$ depends on $a$ and $b$
and can be found by solving a quadratic equation.  
Figure~\ref{fig:ex1:a} illustrates the resulting
partition of the consumer type space into
three regions: those consumers who choose
to buy item $1$, those who choose item $2$,
and those who choose to buy nothing.
Now suppose that in addition to pricing
both items at $p^*$, we also offer a lottery
at price $p^* - \delta$ which yields one of the
two items chosen uniformly at random.  Some of the 
consumers who originally bought items $1$ or $2$,
but were nearly indifferent between them, will
now buy the lottery instead.  This set of 
consumers is represented by the light shaded
area in Figure~\ref{fig:ex1:b}, and each
such consumer pays $\delta$ less than they
would have paid in the original item pricing.
However, this loss of revenue is counterbalanced
by another set of consumers,
represented by the dark shaded area in 
Figure~\ref{fig:ex1:b}, who previously bought
nothing but now pay $p^*-\delta$ for the lottery.
If the values of $a,b,\delta$ are chosen
appropriately, the second effect more than offsets
the first and results in a net increase in revenue.
In effect, the lottery allows for price discrimination
between nearly-indifferent consumers and those who
have a strong preference for one item type.  This
additional price discrimination power allows the
seller to improve revenue.  
In the example discussed here, the seller's net gain is quite
moderate: less than $10\%$ in all of the cases 
analyzed in~\cite{Than}.
Thanassoulis leaves it as an
open question to determine the greatest possible
factor by which the seller can increase revenue
using lottery pricing, noting that it appears 
difficult to resolve this question because of 
the complexity of solving the underlying optimization
problem.

\begin{wrapfigure}{r}{.3\textwidth}
\vspace{-7mm}
  \centering
  \subfigure[Item pricing] 
  {
      \label{fig:ex1:a}
      \epsfig{file=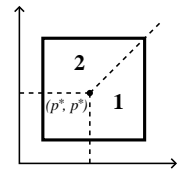,width=25mm}
  }
  \hspace{1cm}
  \subfigure[Lottery pricing] 
  {
      \label{fig:ex1:b}
      \epsfig{file=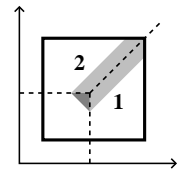,width=25mm}
  }
  \caption{Item pricing versus lottery pricing.}
  \label{fig:sub} 
\vspace{-5mm}
\end{wrapfigure}

\paragraph{Our contributions.}
The preceding discussion raises two obvious questions.
\begin{enumerate}
\item  \label{mq:1}
By what factor can the revenue obtained 
by optimal lottery pricing exceed that obtained 
by optimal item pricing?  Is there any finite
upper bound on this ratio?  If so, how does
the upper bound depend on $n$?
\item  \label{mq:2} 
What is the computational complexity of
evaluating or approximating the optimal revenue 
obtained by lottery pricing? 
\end{enumerate}
We solve both of these questions.  In fact, we
distinguish two versions of the lottery pricing
problem:  the \emph{buy-one} model, in which consumers are only 
allowed to buy one lottery, and the \emph{buy-many}
model, in which they can buy any number of 
lotteries and receive an independent sample 
from each.  It turns out that the answers
to questions~\ref{mq:1} and~\ref{mq:2}
hinge on whether we are in the
buy-one model or the buy-many model.  In the
buy-one model, there is a polynomial-time algorithm
to compute the optimal lottery pricing system,
in contrast to the corresponding item pricing
problem (namely, envy-free unit-demand 
pricing~\cite{envyfree}) which admits no constant-factor 
approximation~\cite{Briest}.  However, there 
is no finite upper bound on the ratio between
optimal lottery revenue and optimal item-pricing
revenue as long as $n>4$.  In the buy-many model,
the ratio is bounded by $\bigo(\log n)$ and this bound
is tight up to a constant factor.  One consequence
is that optimal lottery pricing 
in the buy-many model inherits the inapproximability
of the corresponding item pricing problem.

Our lower bounds on the gap between optimal revenue
with and without lotteries are proven using a variety
of geometric techniques.  In the buy-one model, the
construction relies on the geometry of unit vectors
in Euclidean space, whereas in the buy-many model we
rely on the geometry of degree-$2$ curves in the affine
plane over a finite field.  To obtain the matching upper 
bound we introduce a novel
rounding scheme to transform a lottery pricing system 
into one that only prices pure outcomes while losing
a bounded amount of revenue.  This
type of rounding is quite challenging because, unlike
in the case of rounding fractional LP solutions to integer 
solutions, the objective function is very sensitive to
the choices made during the rounding process: a slight
change in a lottery's probability distribution can cause
a consumer to choose a different lottery and pay a 
different price, resulting in a huge change in revenue.
We overcome this difficulty by coupling the rounding
decisions made for the different items via a single
random variable $t$, then using the properties of 
optimal lottery systems in the buy-many model to bound 
the range of values over which we must sample $t$.

\paragraph{Related work.} As mentioned earlier, randomization is
used extensively in prior-free and online mechanism design (e.g.,
\cite{AGT-HK}, and references therein), as well as for problems where
optimal deterministic truthful mechanisms are hard to compute (e.g.,
\cite{DNS06, APTT03, D3R08}).

In the economics literature, optimal multi-parameter mechanism design and
pricing problems have been studied extensively with a focus towards
deterministic mechanisms (see, e.g., \cite{Armstrong,
  RochetChone}). It is well known~\cite{Mye, RZ} that in
single-parameter Bayesian settings optimal mechanisms are
deterministic, however no general-purpose characterization of optimal
mechanisms is known in the multi-parameter case~\cite{MV-2,
  McAfeeMcMillan}. Thanassoulis~\cite{Than} and Manelli and
Vincent~\cite{MV} independently presented examples showing that in
multiple dimensions randomization can indeed increase the seller's
revenue even when the agents' values are drawn from a product
distribution. The extent of this improvement was unknown prior to our
work.

A number of recent works in CS have explored profit maximization via
``envy-free'' pricing mechanisms~\cite{AFM+04, BB06, envyfree,
  CHK07}. For the setting that we consider, unit-demand bidders with
heterogenous items, Guruswami et al.~\cite{envyfree} presented an
approximation to the optimal envy-free pricing that is logarithmic in
the number of agents. Briest~\cite{Briest} showed that this is
essentially the best possible, under a certain hardness of
approximation assumption for the balanced bipartite independent set
problem, and further that under the same assumption the envy-free
pricing problem is inapproximable to within a factor of $\bigo (n^{\varepsilon})$
for some $\varepsilon >0$. On the positive side, Chawla et al.~\cite{CHK07}
showed that when the values of agents are drawn from a product
distribution, the optimal envy-free pricing can be approximated to
within a factor of $3$.

\section{Preliminaries}
\label{model}

We consider the unit-demand envy-free pricing
problem with $n$ distinct items and some 
distribution $\cd$ on possible consumer
types.  A consumer is given by her
valuation vector $\ct=(v_1,\ldots,v_n) \in
(\R_0^+)^n$ and is interested in purchasing
exactly one of the items.  In the classical
\emph{item pricing problem}, given prices
$p_1,\ldots,p_n$, a consumer chooses to 
purchase the item $i$ maximizing her utility
$v_i-p_i$ or nothing, should this quantity happen
to be negative, and the objective is to find
item prices to maximize the overall revenue.
In the corresponding \emph{lottery pricing
problem}, rather than pricing individual items,
we are allowed to offer an arbitrary system of
lotteries to the consumers.  A lottery $\lambda = (\phi,p)$
is defined by its \emph{probability vector}
$\phi = 
(\phi_1,\ldots,\phi_n), 
\, \sum_{i=1}^n \phi_i \leq 1,$
and by its price $p \in \R_0^+$.  
To define the problem's objective, we need to
specify how consumers select the lottery (or
lotteries) to purchase from a given lottery
system.  This paper considers two alternatives,
which we call the {\em buy-one} and {\em buy-many} models.

\paragraph{The Buy-One Model.}
In the \emph{buy-one model}, we assume that the
consumer picks exactly one lottery from the 
system offered to her.  A consumer type
$\ct=(v_1,\ldots,v_n)$ picking a lottery $\lambda=(\phi,p)$
experiences an expected utility of
$u(\ct,\lambda)=
\left( \sum_{i=1}^n \phi_i v_i \right) - p$, 
i.e.,~the expected valuation for a random
sample from $\phi$, minus the price $p$.
For any consumer type $\ct$ and
set of lotteries $\Lambda$, the utility
maximizing lotteries $\lambda \in \Lambda$
form a set $\Lambda(\ct) = \arg \max_{\lambda \in
\Lambda} \{ u(\ct,\lambda) \}.$  (This set is
well-defined assuming that $\Lambda$ is a closed
subset of $(\R_0^+)^n \times \R_0^+$.  We will
make this assumption throughout the paper.)  
Let 
$
p^+(\ct,\Lambda) = \max \{ p \,|\,
(\phi,p) \in \Lambda(\ct) \}
$
denote the maximum price\footnote{Our assumption 
that consumers choose the highest-priced lottery
in $\Lambda(\ct)$ is without loss of generality,
up to a multiplicative factor of $(1-\varepsilon)$ in
the profit.  This is because we can modify $\Lambda$ 
by decreasing the price of each item by a factor of 
$1-\varepsilon$.  The discount is proportional to
the original price, so the lottery in $\Lambda(\ct)$
that is most attractive to consumer $\ct$ is the
one whose modified price is 
$(1-\varepsilon)p^+(\ct,\Lambda).$  Moreover,
if a lottery outside $\Lambda(\ct)$ is even
more attractive than this one, its discount
must be greater than $\varepsilon p^+(\ct,\Lambda)$
so its modified price must be even higher 
than $(1-\varepsilon p^+(\ct,\Lambda).$
} that a 
utility-maximizing consumer of type $\ct$
might pay when choosing from the lottery set $\Lambda$.
The \emph{profit} 
of $\Lambda$ is defined by 
$
\profit(\Lambda) = \int p^+(\ct,\Lambda) \, d \cd, 
$
where the integral can be expressed as finite weighted
sums provided that $\cd$ has finite support.
The lottery pricing problem in the buy-one 
model asks for a system of lotteries
$\Lambda$ maximizing the profit $\profit(\Lambda).$

\paragraph{The Buy-Many Model.}
While it appears natural in the unit-demand setting
that consumers would choose to purchase a single lottery,
there is a subtle issue to be considered here.  Assume
that consumers can dispose freely of items and consider 
the following simple example with a single item, a single
consumer, and two lotteries.  The consumer values the item
at $1$, lottery $1$ has probability $1$ for allocating the
item and price $1/2$, lottery $2$ has probability $1/2$
and price $2^{-t}$ for some large $t$.  Purchasing lottery
$1$ yields utility $1/2$, while lottery $2$ yields utility
$1/2 - 2^{-t}$, and so the consumer should purchase lottery
$1$ at price $1/2$.  However, imagine the consumer could
decide to purchase $t$ copies of lottery $2$ instead.  The
probability of receiving at least a single item in $t$ 
independent trials is $1-2^{-t}$ and so, under the free
disposal assumption, the resulting utility is $1-(t+1)2^{-t}$,
which is strictly better than $1/2$ for $t \ge 4$.

These considerations motivate the following \emph{buy-many
model}.  Informally, a consumer in the buy-many model 
can purchase any bundle (i.e., multi-set) of lotteries 
and receives an independent sample from each.  Given
this sampling rule, the consumer chooses a utility-maximizing
bundle.  More formally, for an $m$-tuple of lotteries 
$(\lambda_1,\ldots,\lambda_m)$
with probability vectors $\phi_1,\ldots,\phi_m$ and
prices $p_1,\ldots,p_m$, let $(i_1,\ldots,i_m)$
denote an $m$-tuple of independent random variables 
with distributions $\phi_1,\ldots,\phi_m$, respectively.
A consumer $\ct$ purchasing $(\lambda_1,\ldots,\lambda_m)$
experiences a utility of 
$$
u(\ct,\lambda_1,\ldots,\lambda_m) = 
\expect \left[ \max_{1 \leq j \leq m} v_{i_j} \right] - \sum_{j=1}^m p_j.
$$
Note that $u(\ct,\lambda_1,\ldots,\lambda_m)$ does
not depend on the ordering of the sequence $\lambda_1,\ldots,\lambda_m.$
When we refer to the \emph{marginal utility} of adding a 
lottery $\lambda$ to a given bundle $\lambda_1,\ldots,\lambda_m$,
we mean the difference $u(\ct,\lambda_1,\ldots,\lambda_m,\lambda)
- u(\ct,\lambda_1,\ldots,\lambda_m).$

As before, we can define the set $\Lambda^{\mathrm{BM}}(\ct)$ of 
utility-maximizing bundles of lotteries for consumer $\ct$.
This set is well-defined as long as $\Lambda$ is a 
closed subset of $(\R_0^+)^n \times \R^+$.  We can 
define $\profit^{\mathrm{BM}}(\Lambda)$ 
as the revenue achieved when selling to consumer distribution $\cd$ 
assuming every consumer pays for the most expensive bundle of 
lotteries in $\Lambda^{\mathrm{BM}}(\ct).$  
The lottery pricing problem in the buy-many
model asks for a system of lotteries
$\Lambda$ maximizing the profit $\profit^{\mathrm{BM}}(\Lambda)$. Throughout the paper we denote by $r^*(\cd )$ and $r^*_L(\cd )$ the optimal revenue obtainable from consumer distribution $\cd$ via a pure item pricing or a lottery system in the appropriate model.


\section{The Buy-One Model}
\label{buyOne}

We start by considering lottery pricing in the buy-one model. While it turns out that the potential for increased revenue compared to pure item pricings is limited in dimensions $1$ and $2$ (as suggested by our introductory example), surprisingly the situation changes dramatically already in dimension $4$, where an arbitrary increase in revenue is possible. Remarkably, if the consumer distribution has finite support, computing the optimal lottery system in the buy-many model reduces to solving a linear program, thereby circumventing the known hardness results for item pricing.

\subsection{A Polynomial-Time Algorithm}
\label{polytimeAlgo}

\begin{wrapfigure}{r}{.55\textwidth}
  \vspace{-7mm}
  \begin{eqnarray}
  \mbox{max.} & & \sum _{j=1}^m \mu _jz_j\\
  \label{C1} \mbox{s.t.} & & \sum _{i=1}^n x_{ji}\le 1 \quad \quad \quad \quad \quad \quad \quad \quad \quad \quad \forall j\\
  \label{C2} & & \sum _{i=1}^n x_{ji}v_{ji}-z_j\ge 0 \quad \quad \quad \quad \quad \quad \, \, \, \, \forall j\\
  \label{C3} & & \sum _{i=1}^n x_{ji}v_{ji}-z_j\ge \sum _{i=1}^n x_{ki}v_{ji}-z_k \quad \forall j,k
  \end{eqnarray}
  \vspace{-12mm}
\end{wrapfigure}

Let $\mathcal{C}$ be a finite support distribution on consumer types $v_j=(v_{j1},\ldots ,v_{jn})$ with probabilities $\mu _j$ for $1\le j\le m$, and consider the LP to the right (omitting non-negativity constraints $x_{ji},z_j\ge 0$),
where $x_j=(x_{j1},\ldots ,x_{jn})$ is the probability vector of a lottery offered at price $z_j$. Constraints (\ref{C1}) ensure that lotteries are feasible. Constraints (\ref{C2}) guarantee that consumers of type $v_j$ can afford to buy lottery $x_j$. Constraints (\ref{C3}) ensure that consumer $v_j$ prefers lottery $x_j$ over all other $x_k$.

\begin{theorem}
\label{t:optRev}
For consumers specified as a finite support distribution the optimal lottery system in the buy-one model can be computed in polynomial time.
\end{theorem}

\subsection{Lotteries in Dimensions $1$ and $2$}
\label{dim1and2}

Given a consumer distribution $\cd$, by how much can the optimal lottery revenue $r^*_L(\cd )$ exceed the optimal item pricing revenue $r^*(\cd )$? We first consider the base case of only a single item and prove that offering lotteries cannot yield higher revenue than the best single item price.

\begin{theorem}
\label{t:dim1}
Let $n=1$ and $\mathcal{C}$ be a consumer distribution. It holds that $r^*_L(\mathcal{C})=r^*(\mathcal{C})$.
\end{theorem}
\begin{proof}
Let $\Lambda=\{ (\phi _0,p_0),(\phi _1,p_1),\ldots ,(\phi _m,p_m)\}$ denote the optimal lottery system and assume w.l.o.g. that $0=\phi _0<\phi _1<\cdots <\phi _m$ and $0=p_0\le p_1\le \cdots \le p_m$. Now consider a consumer with valuation $v$ who decides to purchase lottery $(\phi _j,p_j)$. Since $(\phi _j,p_j)$ is the utility maximizing choice it must be the case that $\phi _jv-p_j\ge \phi _{j-1}v-p_{j-1}$ and $\phi _jv-p_j> \phi _{j+1}v-p_{j+1}$ and, thus,\[
v\in \left[\frac{p_j-p_{j-1}}{\phi _j-\phi _{j-1}},\frac{p_{j+1}-p_j}{\phi _{j+1}-\phi _j}\right].\]
Assume that we randomly assign price $(p_j-p_{j-1})/(\phi _j-\phi _{j-1})$ with probability $\phi _j-\phi _{j-1}$ to the item for $1\le j\le m$. For the payment $P$ made by the consumer with value $v$ we may write that\[
\expect [P] = \sum _{i=1}^j (\phi _i-\phi _{i-1})\frac{p_i-p_{i-1}}{\phi _i-\phi _{i-1}} = \sum _{i=1}^j (p_i-p_{i-1}) = p_j-p_0=p_j,\]
which is just what she pays given the optimal lottery system.
\end{proof}

Looking at the case of $2$ distinct items and only allowing lotteries that have total probability $1$ of allocating either of the two items, it is possible to still derive a constant bound on the ratio between the optimal lottery and item pricing revenue. The proof, which is found in Appendix \ref{a:dim2}, is based on a combination of geometric arguments and a reduction to the $1$-dimensional case above. In dimension $2$, every lottery corresponds to an {\em indifference line}, which is the set of consumer valuations that would result in zero utility from buying this lottery (the {\em indifferent} consumers). Fig. \ref{fig:2dim} depicts a lottery system in dimension $2$ and the corresponding item pricing constructed in the proof of Theorem \ref{t:dim2}.

\begin{theorem}
\label{t:dim2}
For $n = 2$ it holds that $r^*_L(\cd ) \leq 3r^*(\cd )$ for any consumer distribution $\cd$.
\end{theorem}

\pagebreak

\subsection{Higher Dimensions}
\label{higherDims}

\begin{wrapfigure}{l}{.3\textwidth}
  \vspace{-7mm}
  \begin{center}
    \epsfig{file=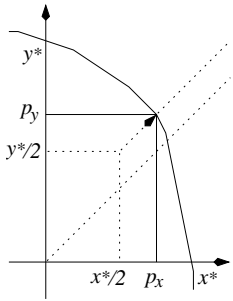,width=30mm}
    \caption{\label{fig:2dim} Pricing in dimension $2$.}
  \end{center}
  \vspace{-7mm}
\end{wrapfigure}

The results in Theorems \ref{t:dim1} and \ref{t:dim2} suggest that it should be possible to derive a general bound for the revenue gap between lottery and item pricing depending on the problem dimension in some way. In fact, consider the special case of consumers with {\em uniform valuations}, i.e., each consumer has value $v$ for all items in some set $S$ and value $0$ for any item from the complement of $S$. Grouping consumers according to the set $S$ they are interested in and applying the randomized rounding technique from Theorem \ref{t:dim2} (using for each lottery in the system its probability of allocating an item from $S$), we obtain an upper bound of $\bigo (2^n)$ on the revenue gap and this turns out to be essentially tight. The proof of Theorem \ref{t:uniVals} is found in Appendix \ref{a:uniVals}.

\begin{theorem}
\label{t:uniVals}
Let $\cd$ be a distribution on uniform valuation consumers. Then $r^*_L(\cd )/r^*(\cd )=\bigo (2^n)$. There exist distributions $\cd$ with $r^*_L(\cd )/r^*(\cd )=\tilde{\Omega }(2^n)$.
\end{theorem}

Quite surprisingly, a similar result does not hold for general consumer distributions. We show that the difference in revenue between the optimal item pricing and the best lottery system cannot be bounded in terms of the number of items and, in fact, the gap can become arbitrarily large already in dimension $4$.

\begin{theorem}
\label{t:unbounded}
For any number of items $n\ge 4$ the maximum gap $r^*_L(\mathcal{C})/r^*(\mathcal{C})$ between the revenue obtainable by a system of lotteries in the buy-one model and the optimal item pricing cannot be bounded as a function of $n$.
\end{theorem}

{\bf Proof of Theorem \ref{t:unbounded}:}
We view lotteries and valuations as normalized vectors in the all-positive orthant of $\R ^n_+$. The main observation is that valuation vectors with some minimum distance from each other allow for price discrimination among these consumers via carefully chosen lotteries. Our construction relies essentially on the following technical lemma, which gives a lower bound on the number of valuation vectors we can pack without violating our minimum distance constraint. Let $S^n_r$ be the $n$-dimensional sphere with radius $r$ centered at the origin. By $S^{n+}_r$ we refer to its intersection with the all-positive orthant $\R ^n_+$. The proof of Lemma \ref{t:vectorPacking} is found in Appendix \ref{a:vectorPacking}.

\begin{lemma}
\label{t:vectorPacking}
Let $n\ge 1$ be given. For every $q\ge 2n$ there exists a set $\mathcal{V}^n_q$ of vectors in $S^{n+}_{1/\sqrt{n}}$, such that $v\cdot w\le 1/n-1/q$ for all $v,w\in \mathcal{V}^n_q$ with $v\not= w$ and $|\mathcal{V}^n_q|=\Omega (q^{(n-1)/2})$.
\end{lemma}

For a given choice of $n$ and $q$ let now the set of vectors $\mathcal{V}^n_q=\{ v_1,\ldots ,v_{(\ell (q))}\}$ with $\ell (q)=\Omega (q^{(n-1)/2})$ and $v_i\cdot v_j\le 1/n-1/q$ for any $i\not= j$ be given. We will define a unit-demand pricing instance based on these vectors. Consumer distribution $\cd$ will be defined as a finite support distribution. For each $v_j$ we define a consumer type with valuation vector $\tilde{v}_j=2^j\cdot v_j$ and probability $\mu _j=2^{-j}$. Since $||v_j||_1 \le \sqrt{n}||v_j||_2=1$, vectors from $\mathcal{V}^n_q$ can also be interpreted as lotteries. We define lotteries $\lambda _j$ with probability vectors $\phi _j=v_j$ and assign them price $p_j=(1/q)\cdot 2^j$.

Consider the utility $u(\tilde{v}_j,\lambda _j)$ of consumer type $\tilde{v}_j$ when purchasing lottery $\lambda _j$. We may write that $u(\tilde{v}_j,\lambda _j) = \tilde{v}_j\cdot \phi _j-(1/q)\cdot 2^j = 2^j(v_j)^2 -(1/q)\cdot 2^j=(1/n-1/q) 2^j$, by the fact that $||v_j||_2=1/\sqrt{n}$. On the other hand, the consumer type's utility from any other lottery $\lambda _i$ is bounded above by $u(\tilde{v}_j,\lambda _i) = \tilde{v}_j\cdot \phi _i=\left( v_j\cdot v_i\right) 2^j \le (1/n-1/q) 2^j$, since $v_j\cdot v_i\le 1/n-1/q$ for all $v_j,v_i\in \mathcal{V}^n_q$. Thus, given the lottery system defined above each consumer $\tilde{v}_j$ will choose to purchase lottery $\lambda _j$ and we obtain revenue $\sum _{j=1}^{\ell (q)}2^{-j}\cdot (1/q)\cdot 2^j=\ell (q)/q=\Omega (q^{(n-3)/2})$, using that $\ell (q)=\Omega (q^{(n-1)/2})$. It remains to estimate the optimal item pricing revenue. Consider a single item priced at $p\in \R _+$ and all other items priced at $+\infty$. For $2^{k-1}<p\le 2^k$ consumer types $\tilde{v}_1,\ldots ,\tilde{v}_{k-1}$ surely have valuations of less than $p$ for all items. It follows that the total probability mass of consumers who are able to afford the item is bounded above by $\sum _{j=k}^{\ell (q)}2^{-j}\le 2^{-k+1}$ and, thus, total revenue is at most $p\cdot 2^{-k+1}\le 2$. It follows that the optimal item pricing results in revenue of at most $2n$ and for any constant $n$ we obtain a lower bound of $\Omega(q^{(n-3)/2})$ on the revenue gap. In particular, we can make this gap arbitrarily large by choosing a sufficiently large $q$ in any dimension $n\ge 4$. \hfill $\square$

\section{The Buy-Many Model}
\label{buyMany}

As we have argued before, the assumption made in the buy-one model that a consumer purchases a single utility maximizing lottery from any given system is not justifiable in general. We now continue by investigating the more realistic buy-many model, in which we allow consumers to buy any combination of lotteries maximizing their expected utility. As we will see in Section \ref{logBound}, this reduces the advantage lottery systems have over pure item pricings drastically. In particular, this implies that known inapproximability results for the item pricing problem yield similar bounds in the lottery setting and algorithmic results similar to the buy-one model cannot be obtained. In Section \ref{logLower} we prove that our bound on the revenue gap is asymptotically tight.

\subsection{Upper Bound and Hardness of Approximation}
\label{logBound}

Let an arbitrary system of lotteries in the buy-many model over $n$ distinct items be given. We assume throughout this section that the utility maximizing collection of lotteries for each consumer type given this system consists of a single lottery. This assumption is w.l.o.g., as we can add lotteries corresponding to the joint distribution of some collection of lotteries to the system until it holds. The following randomized algorithm turns the lottery system into a pure item pricing:
\begin{itemize}
\item[(1)] For each item $i$, let $p_i$ be the price of the cheapest lottery with probability at least $1/(130n^3)$ for item $i$ ($p_i=+\infty$ if no such lottery exists).
\item[(2)] With probability $1/2$, uniformly sample $t$ from $\{ -1,0,\ldots ,3\lfloor \log n\rfloor +9\}$ and assign price $2^tp_i$ to every item $i$.
\item[(3)] Else sample a single item $i$ uniformly at random. Assign price $+\infty$ to all items other than $i$. Price item $i$ at $130n^3\ee ^jp_i$ with probability $(1-1/\ee )\ee ^{-j}$ for all $j\in \N _0$.
\end{itemize}

The core idea of the algorithm is the following: Every lottery with some minimum probability of allocating some specific item defines an upper bound on the payment of consumer types prefering this item, since by buying multiple copies of the lottery at hand they can make the probability of receiving the desired item approach $1$ exponentially fast. Thus, for each item we let the cheapest lottery with some minimum probability for it define its {\em base price} and assign a random item price via a carefully tailored two stage stochastic process. We are going to argue that the above algorithm outputs an item pricing that is an expected $\mathcal{O}(\log n $)-approximation to the revenue of an optimal lottery system in the buy-many model. Note, that the algorithm is easily derandomized via exhaustive search over the entire range of (relevant) random coin flips.  Throughout this section we assume w.l.o.g. that $\sum _{i=1}^n \phi _i = 1$ for every lottery. This is easily achieved by adding a dummy item valued at $0$ by all consumers to the instance.

\begin{theorem}
\label{t:logUpper}
Given a distribution $\mathcal{C}$ of unit-demand consumers and an optimal lottery system in the buy-many model, the above algorithm returns an item pricing with revenue $r\ge 1/\mathcal{O}(\log n)\cdot r^*_{L}(\mathcal{C})$.
\end{theorem}

It is known that the unit-demand item pricing problem cannot be approximated within $\bigo (n^{\varepsilon})$ for some $\varepsilon >0$. Thus, we immediately obtain the following inapproximability result for lottery pricing in the buy-many model.

\begin{corollary}
\label{t:inapprox}
The unit-demand lottery pricing problem cannot be approximated within $\bigo (n^{\varepsilon})$ for some $\varepsilon >0$, unless NP $\subseteq \bigcap _{\delta >0}$ BPTIME$(2^{\bigo (n^{\delta})})$.
\end{corollary}

{\bf Proof of Theorem \ref{t:logUpper}:}
We will show that in going from the optimal lottery system to a pure item pricing, the expected loss in revenue is bounded by $\mathcal{O}(\log n)$ for every single consumer type. So let a single consumer type from $\mathcal{C}$ with values $(v_1,\ldots ,v_n)$ be given. Furthermore, assume that this consumer buys a lottery with probabilities $(\phi _1,\ldots ,\phi _n)$ and price $p$ when offered the optimal lottery system in the buy-many model. By $j=\argmax \{ v_i\, |\, i\, :\, \phi _i \ge 1/(130n^3)\}$ we denote the consumer's favorite item among those for which he has at least a $1/(130n^3)$-chance of receiving them.

Finally, recall that for every item $i$, $p_i$ denotes the price of the cheapest lottery that has probability at least $1/(130n^3)$ for item $i$. We start with the observation that in any optimal lottery system a consumer's utility does not depend significantly on items she receives with negligible probability.

\begin{proposition}
\label{prop:logUpper}
We may w.l.o.g. assume that $\sum_{i=1}^{n}{\phi _i(v_i-v_j)} \le p/4$.
\end{proposition}
\begin{proof}
Let $\mathcal{C}'$ be the set of all consumer types for which the above does not hold. In particular, for each consumer type in $\mathcal{C}'$, if we let $j$ again denote her favorite item with probability at least $1/(130n^3)$ in the lottery $(\phi _1,\ldots ,\phi _n)$ she buys at price $p$, there exists an item $k$ with $\phi _k(v_k-v_j)>p/(4n)$. Let $\mathcal{C}'_k$ be the set of all consumer types for which item $k$ satisfies this inequality and consider a single class $\mathcal{C}'_k$.

Remove all lotteries except for the ones bought by consumers in $\mathcal{C}'_k$. For the remaining lotteries, set their probabilities for all items other than $k$ to $0$ and reduce their prices by a factor of  $8n$. Now consider a single consumer type in $\mathcal{C}'_k$ with values $(v_1,\ldots ,v_n)$ buying lottery $(\phi _1,\ldots ,\phi _n)$ at price $p$ in the original lottery system and favorite item $j$ among those with minimum probability $1/(130n^3)$. We want to lower bound the revenue from this consumer given the modified system of lotteries.

If the consumer purchases the modified version of the lottery she bought originally, revenue has decreased by at most a factor of $8n$. If she does not, she now buys some other lottery (or a combination of lotteries) with some probability $\mu _k$ for item $k$ and price $q/8n$. With lottery $(\phi _1,\ldots ,\phi _n)$, the consumer had a chance of at least $1-(n-1)/(130n^3)\ge 1-1/(130n^2)$ of receiving an item valued at $v_j$ or less. Thus, the marginal utility of adding a copy of $(\mu _1,\ldots ,\mu _n)$ at its original price $q$ would have been $\mu _k(1-1/(130n^2))(v_k-v_j)-q$. Since the consumer chooses not to buy a copy, we have\[
\mu _k\left( 1-\frac{1}{130n^2}\right)(v_k-v_j)-q\le 0.\]
We know that $\phi _kv_k\ge \phi _k(v_k-v_j)\ge p/(4n)$ and, thus, buying the modified version of lottery $(\phi _1,\ldots ,\phi _n)$ at price $p/(8n)$ yields utility $\phi _kv_k-p/(8n)\ge \frac{1}{2}\phi _kv_k$. Thus, it must be the case that $\mu _k\ge \phi _k/2$. Finally, we obtain
\begin{eqnarray*}
q & \ge & \mu _k\left( 1-\frac{1}{130n^2}\right)(v_k-v_j) \ge \frac{\phi _k}{2}\left( 1-\frac{1}{130n^2}\right)(v_k-v_j)\\
 & \ge & \frac{1}{8n}\left( 1-\frac{1}{130n^2}\right)p.
\end{eqnarray*}
Thus, if the consumer purchases at price $q/8n$ the reduction in revenue is bounded below by $1/(64n^2)(1-1/(130n^2))\ge 1/(65n^2)$. Now observe that all lotteries in our modified lottery system have probability at most $1/(130n^3)$, so we can multiply probabilities and prices of all lotteries by $130n^3$ without affecting any consumer's buying decision. This effectively increases the revenue from every consumer in $\mathcal{C}'_k$ by a factor of $(130n^3)/(65n^2)=2n$ compared to the original optimal lottery system.

Now assume that more than half the revenue of the original optimal lottery system was due to consumer types in $\mathcal{C}'$. Then there must exist a class $\mathcal{C}'_k$ that carries more than a $1/(2n)$-fraction of the overall revenue, which we have just shown how to increase by a factor of $2n$, a contradiction. Hence, at most half the revenue stems from consumer types in $\mathcal{C}$ and we may ignore these consumer types.
\end{proof}

Consider the price assignments defined in Step (2) of the algorithm for different values of $t\in \{ -1,0,\ldots ,3\lfloor \log n\rfloor +9\}$. We know that $v_j=\sum _{i=1}^n\phi _iv_i-\sum _{i=1}^n\phi _i(v_i-v_j)\ge p-p/4=(3/4)p$ and, thus, for $t=-1$ the consumer has strictly positive utility from buying $j$ and will consequently purchase some item. Denote by $i_0$ the item bought for $t=-1$. For increasing values of $t$, the consumer might switch to other items that yield higher utility. Refer to these items as $i_1,\ldots ,i_{\ell}$ in the order they are bought.

\begin{fact}
\label{lem:monotonicity}
It holds that $p_{i_0}>p_{i_1}>\cdots >p_{i_{\ell}}$ and $v_{i_0}\ge v_{i_1}\ge \cdots \ge v_{i_{\ell}}$.
\end{fact}
\begin{proof}
Note that the consumer buys only if this results in non-negative utility. In going from $t$ to $t+1$ the utility from buying any item $i$ decreases by $(2^{t+1}-2^t)p_i=2^tp_i$ and, thus, decreases strictly less on cheaper items. It follows that $p_{i_0}>p_{i_1}>\cdots >p_{i_{\ell}}$.

Assume then that $p_{i_j}>p_{i_{j+1}}$, but $v_{i_j}<v_{i_{j+1}}$ for some $i_j$, $i_{j+1}$. Then buying item $i_{j+1}$ yields strictly higher utility than $i_j$ for any value of $t$ and $i_j$ is never bought.
\end{proof}

We proceed by deriving a lower bound on the consumer's utility given the original lottery system. For $t=-1$, her utility from the item pricing is at least $v_j-p_j/2\ge v_j-p/2$, since $\phi _j\ge 1/(130n^3)$ and, thus, $p_j\le p$ by definition. Since she decides to purchase item $i_0$, it must be the case that $v_{i_0}\ge v_j-p/2$. We may then write that
\begin{eqnarray}
\sum _{i=1}^n\phi _iv_i -p & = & \sum _{i=1}^n\phi _i(v_i-v_j)+v_j -p\le v_j-\frac{3}{4}p \quad \quad \mbox{ , since } \sum _{i=1}^n\phi _i(v_i-v_j)\le \frac{p}{4}\\
\label{eqn1} & \le & v_{i_0}-\frac{p}{4} \quad \quad \quad \quad \quad \quad \quad \quad \quad \quad \quad \quad \quad \, \, \, \mbox{ , since } v_j\le v_{i_0}+\frac{p}{2}.
\end{eqnarray}

Next, we are going to bound the utility the consumer can achieve by focusing on item $i_{\ell}$ from below. Let $(\mu _1,\ldots ,\mu _n)$ be the lottery defining price $p_{i_{\ell}}$, i.e., the cheapest lottery in the original system with $\mu _{i_{\ell}}\ge 1/(130n^3)$. Now consider a strictly worse lottery with probability exactly $1/(130n^3)$ for item $i_{\ell}$, probability $0$ for all other items and price $p_{i_{\ell}}$. Assume that this was the only available lottery and let $k$ denote the number of copies our consumer would choose to purchase. Then, by the fact that the $(k+1)$-th copy does not yield positive marginal utility for her, we may conclude that $1/(130n^3)( 1-1/(130n^3)) ^kv_{i_{\ell}}-p_{i_{\ell}}\le 0$ and, thus, $(1-1/(130n^3))^kv_{i_{\ell}}\le 130n^3p_{i_{\ell}}$. Consequently, the utility from buying $k$ copies is at least
\begin{eqnarray}
\label{eqn2} \left( 1-\Bigl( 1-\frac{1}{130n^3}\Bigr) ^k\right) v_{i_{\ell}}-kp_{i_{\ell}} & \ge & \max \Bigl\{ v_{i_{\ell}}-(k+130n^3)p_{i_{\ell}}\, ,\, 0 \Bigr\} .
\end{eqnarray}
By the fact that a strictly better lottery was part of the original lottery system, combining (\ref{eqn1}) and (\ref{eqn2}) yields\[
v_{i_0}-\frac{p}{4}\ge \sum _{i=1}^n\phi _iv_i -p \ge \max \Bigl\{ v_{i_{\ell}}-(k+130n^3)p_{i_{\ell}}\, ,\, 0 \Bigr\} .\]
Finally, rearranging for $p$ we obtain\[
p\le 4\left( v_{i_0}-v_{i_{\ell}}+\min \Bigl\{ (k+130n^3)p_{i_{\ell}}\, , \, v_{i_{\ell}}\Bigr\}\right)\]
as an upper bound on the price paid by the consumer given the original lottery system.

We proceed by proving a lower bound on the expected price paid given the item pricing returned by our algorithm. We distinguish the following three cases.

{\bf Case (1):} $v_{i_{\ell}}\le (k+130n^3)p_{i_{\ell}}$. Let $t_j$ be the highest value of $t$ in Step (2) of the algorithm at which the consumer purchases item $i_j$. Observe that as long as item $i_j$ is priced at $v_{i_j}-v_{i_{j+1}}$ or less it yields higher utility than item $i_{j+1}$. It follows that\[
2^{t_j+1}p_{i_j}\ge v_{i_j}-v_{i_{j+1}} \mbox{\hspace{3mm} or, equivalently, \hspace{3mm}} 2^{t_j}p_{i_j}\ge (1/2)(v_{i_j}-v_{i_{j+1}}).\]
Using that $v_{i_{\ell}}\le (k+130n^3)p_{i_{\ell}}$ and $v_{i_{\ell}}>130n^3(1-1/(130n^3))^{-(k-1)}p_{i_{\ell}}$ (by the fact that the $k$-th copy of the lottery with probability $1/(130n^3)$ for item $i_{\ell}$ at price $p_{i_{\ell}}$ has positive marginal utility), it readily follows that $v_{i_{\ell}}\le 260n^3p_{i_{\ell}}$. In particular, this implies that the range of $t$ in Step (2) of the algorithm includes $\lfloor \log v_{i_{\ell}}\rfloor$. Let now $R$ denote the the price paid by the consumer given the item pricing and note that Step (2) of the algorithm is performed with probability $1/2$. We have
\begin{eqnarray*}
\expect [R] & \ge & \frac{1}{2}\left( \sum _{j=0}^{\ell -1}\prob (t=t_j)\frac{1}{2}(v_{i_j}-v_{i_{j+1}}) + \prob (t=\lfloor \log v_{i_{\ell}}\rfloor )\frac{v_{i_{\ell}}}{2} \right)\\
 & \ge & \frac{1}{12\lfloor \log n\rfloor +44}\left( \sum _{j=0}^{\ell -1}(v_{i_j}-v_{i_{j+1}})+v_{i_{\ell}}\right) = \frac{1}{12\lfloor \log n\rfloor +44}v_{i_0}.
\end{eqnarray*}

{\bf Case (2):} $v_{i_{\ell}}>(k+130n^3)p_{i_{\ell}}$, $k\le 130n^3+2n$. In this case we know that $\lfloor \log k+130n^3\rfloor$ lies within the range of $t$ in Step (2) of the algorithm. Similar to Case (1) above we obtain\[
\expect [R]\ge \frac{1}{12\lfloor \log n\rfloor +44}\Bigl(v_{i_0}-v_{i_{\ell}}+(k+130n^3)p_{i_{\ell}}\Bigr) .\]

{\bf Case (3):} $v_{i_{\ell}}>(k+130n^3)p_{i_{\ell}}$, $k> 130n^3+2n$. Once more, recall that our consumer has positive marginal utility from buying the $k$-th copy of a lottery that offers a $1/(130n^3)$-chance of receiving item $i_{\ell}$ at price $p_{i_{\ell}}$. Thus,\[
v_{i_{\ell}}\ge 130n^3\Bigl( 1-\frac{1}{130n^3}\Bigr) ^{-(k-1)}p_{i_{\ell}}\approx 130n^3\ee ^{k-130n^3-1}p_{i_{\ell}},\]
where the above holds with arbitrary precision for large values of $n$. Let $A_{i_{\ell}}$ denote the event that the algorithm chooses to perform the random experiment in Step (3) and picks item $i_{\ell}$ to assign a price different from $+\infty$ to. For the expected payment of our consumer conditioned on $A_{i_{\ell}}$ we have
\begin{eqnarray*}
\expect [R\, |\, A_{i_{\ell}}] & \ge & \sum _{j=0}^{k-130n^3-1}\Bigl( 1-\frac{1}{e} \Bigr)e^{-j}130n^3e^jp_{i_{\ell}} = (k-130n^3)\Bigl( 1-\frac{1}{e} \Bigr) 130n^3p_{i_{\ell}}\\
 & = & (nk+(130n^3-n)k-130^2n^6)\Bigl( 1-\frac{1}{e} \Bigr) p_{i_{\ell}}\ge \Bigl( 1-\frac{1}{e} \Bigr) nkp_{i_{\ell}}.
\end{eqnarray*}
Finally, since event $A_{i_{\ell}}$ has probability $1/(2n)$, we obtain
\begin{eqnarray*}
\expect [R] & \ge & \frac{1}{2}\left( \sum _{j=0}^{\ell -1}\prob (t=t_j)\frac{1}{2}(v_{i_j}-v_{i_{j+1}})\right) + \frac{1}{2n}\expect [R\, |\, A_{i_{\ell}}]\\
 & \ge & \frac{1}{12\lfloor \log n\rfloor +44}(v_{i_0}-v_{i_{\ell}})+\frac{\ee -1}{2\ee}kp_{i_{\ell}} \ge \frac{1}{12\lfloor \log n\rfloor +44}\Bigl( v_{i_0}-v_{i_{\ell}}+(k+130n^3)p_{i_{\ell}}\Bigr) .
\end{eqnarray*}
So we have that $\expect [R]= 1/\mathcal{O}(\log n)\cdot p$ in each case, which finishes the proof. \hfill$\square$

\subsection{A Lower Bound}
\label{logLower}

Finally, we will show that the bound derived in the previous section is tight. In fact, it turns out that this is even true for the restricted case of consumers with uniform valuations. This is somewhat surprising, as this distinction is quite significant in the buy-one model, as we have seen before.

\begin{theorem}
\label{t:uniformIndGap}
For all $n$ there exist uniform-valuation consumer distributions $\mathcal{C}$ with $r^*_L(\mathcal{C})/r^*(\mathcal{C})=\Omega (\log n)$.
\end{theorem}

{\bf Proof of Theorem \ref{t:uniformIndGap}:}
Let $n\in \N$ be prime and $\mathcal{P}$ denote the set of all distinct polynomials of degree $2$ over the field $\Z / n\Z$. We identify each polynomial $P\in \mathcal{P}$ with the set $S_P=\bigl\{ (0,P(0)),\ldots ,(n-1,P(n-1)) \bigr\} \subset \bigl( \Z /n\Z \bigr) ^2$. Observe that $|\mathcal{P}|=n^3$ and $|S_P|=n$, $|S_P\cap S_Q|\le 2$ for all $P\not= Q\in \mathcal{P}$ by the fact that polynomials of maximum degree $2$ over $\Z / n\Z$ for $n$ prime have at most $2$ zeroes.

We define a random pricing instance based on these polynomials as follows. The set of items corresponds to the elements of $( \Z /n\Z )^2$. Let $k=\lfloor \log n \rfloor$. For each $P\in \mathcal{P}$ we define a corresponding class $\mathcal{C}_P$ of $2^{k-j}$ identical consumers with a non-zero value $v_P=2^{j-k}$ for the items in $S_P$, where $j$ is drawn uniformly at random from $\{ 1,\ldots ,k\}$. Let $\mathcal{C}=\bigcup _P \mathcal{C}_P$ denote the complete instance.

The proof of the lower bound proceeds in two steps. We first argue that with non-zero probability our random experiment creates a pricing instance on which every pure item pricing yields revenue $\mathcal{O}(n^3/\log n)$. We then show that for every instance created by the experiment we can find a lottery system that yields revenue $\Omega (n^3)$.

\begin{lemma}
\label{l:uniformItemBound}
Let $\mathcal{C}$ be a random pricing instance as defined above. It holds that $r^*(\mathcal{C})=\mathcal{O}(n^3/\log n)$ with positive probability.
\end{lemma}
\begin{proof}
Let $p$ be a price vector resulting in revenue $r$ on some instance created by our random experiment. Then there must exist another price vector $p'$ that assigns only prices from $\{ 2^{1-k},2^{2-k},\ldots ,2^{0}\}$ and makes revenue at least $r/2$ by sales to consumers buying at a price equal to their value. To see this, note that it is w.l.o.g. to assume that prices are powers of $2$ and as long as more than half the revenue generated by a price vector $p$ comes from consumers buying at a price no more than half their values, we can increase revenue by doubling all prices. Thus, since the identical consumers in each class $\mathcal{C}_P$ contribute a total revenue of $1$ if they buy at their full values, we only need to prove that with positive probability no price vector extracts the full value of consumers in more than $\mathcal{O}(n^3/ \log n)$ different classes $\mathcal{C}_P$.

Consider a fixed price vector $p$ and let $T_j$ denote the set of items priced at $2^{j-k}$ for $j=1,\ldots ,k$. By $\mathcal{B}_j$ we denote the set of consumer classes with a non-zero value for at least one item in $T_j$ and value zero for all items in $T_1,\ldots ,T_{j-1}$, formally,\[
\mathcal{B}_j=\left\{ \mathcal{C}_P\, |\, S_P\cap T_j \not= \emptyset \right\} \backslash \Bigl( \bigcup _{i=1}^{j-1}\mathcal{B}_i\Bigr).\]
A set $\mathcal{C}_P\in \mathcal{B}_j$ of consumers yields revenue $1$ if and only if their random value is $v_P=2^{j-k}$. Let random variable $Y_j$ denote the number of consumer classes in $\mathcal{B}_j$ that pay their full values given price vector $p$. We will bound $Y_j$ in two steps. If $|\mathcal{B}_j|\le n^3/(\log n)^2$, then it trivially holds that $Y_j\le |\mathcal{B}_j|\le n^3/(\log n)^2$. If $|\mathcal{B}_j|> n^3/(\log n)^2$, let $X_P^j$ be a random indicator variable with $X_P^j=1$ if $v_P=2^{j-k}$ and $X_P^j=0$ else. Thus, $\prob (X_P^j=1)=1/k$. It follows that that\[
Y_j=\sum _{P:\mathcal{C}_P\in \mathcal{B}_j}X_P^j \mbox{\hspace{3mm} and \hspace{3mm}} \expect [Y_j]=\frac{1}{k}|\mathcal{B}_j|.\]
Applying the Chernoff bound \cite{MitzenmacherUpfal} and using that $|\mathcal{B}_j|> n^3/(\log n)^2$ and $k\le \log n$ we obtain\[
\prob \left( Y_j\ge \frac{6}{k}|\mathcal{B}_j|\right) \le 2^{-\frac{6}{k}|\mathcal{B}_j|} \le 2^{-6\left( \frac{n}{\log n}\right) ^3}.\]
Thus, we have that $Y_j\le \max \{ n^3/(\log n)^2,12|\mathcal{B}_j|/(\log n) \}$ with probability at least $1-2^{-6\left( \frac{n}{\log n}\right) ^3}$. Let $Y=\sum _j Y_j$. Applying the union bound yields\[
Y = \sum _{j=1}^k Y_j \le \sum _{j=1}^k \left( \frac{n^3}{(\log n)^2}+\frac{12|\mathcal{B}_j|}{\log n}\right) = k\frac{n^3}{(\log n)^2}+\frac{12\sum _{j=1}^k |\mathcal{B}_j|}{\log n} \le \frac{13n^3}{\log n}\]
with probability at least $1-2^{-5\left( \frac{n}{\log n}\right) ^3}$. Finally, observe that there are at most $k^{n^2}$ different price vectors with prices in $\{ 2^{1-k},\ldots ,2^{0}\}$. Applying the union bound once again the probability that any of these extracts full values from more than $13n^3/\log n$ consumer classes is at most $(\log n)^{n^2}2^{-5\left( \frac{n}{\log n}\right) ^3}<1$ for $n\ge 2$.
\end{proof}

\begin{lemma}
\label{l:uniformLotteryBound}
Let $\mathcal{C}$ be any pricing instance as defined above. It holds that $r_L^*(\mathcal{C})=\Omega (n^3)$.
\end{lemma}
\begin{proof}
We construct a system of lotteries as follows. For every consumer class $\mathcal{C}_P$ with value $v_P$ for items in set $S_P$, we introduce a lottery $\phi _P$ with probability $1/n$ for each of the items in $S_P$ at price $v_P/2$. A consumer from class $\mathcal{C}_P$ has utility $v_P-(1/2)v_P=(1/2)v_P$ from purchasing lottery $\phi _P$. By the fact that $|S_P\cap S_Q|\le 2$ for all $Q\not= P$, any other lottery $\phi _Q$ has probability of at most $2/n$ of allocating an item in set $S_P$. We also know that, since $k\le \log n$, we have $v_Q\ge 2^{-k}v_P\ge (1/n)v_P$ and, thus, a consumer from $\mathcal{C}_P$ has utility at most $(2/n)v_P-(1/2)v_Q\le (2/n)v_P-(1/(2n))v_P=(3/(2n))v_P$. Since the same bounds the marginal utility of any lottery other than $\phi _P$ being bought in addition to some other set of lotteries, it follows that consumers in $\mathcal{C}_P$ either purchase lottery $\phi _P$ at price $v_P/2$ or some set of at least $((1/2)v_P)/((3/(2n))v_P)=(1/3)n$ lotteries at price at least $(1/n)v_P$ each, resulting in overall revenue at least $(1/3)v_P$ from each of these consumers. Summing over all classes $\mathcal{C}_P$ and the consumers in each class yields overall revenue of $(1/3)n^3$.
\end{proof}

\paragraph{Acknowledgements}
We are very grateful to Jason Hartline for numerous enlightening discussions in the early stages of this research. The result in Theorem \ref{t:dim2} has independently been discovered by David Malec.

\bibliographystyle{plain}
\bibliography{lotteryPricing}

\begin{thebibliography}{10}

\bibitem{AFM+04}
G.~Aggarwal, T.~Feder, R.~Motwani, and A.~Zhu.
\newblock Algorithms for multi-product pricing.
\newblock In {\em Proc. 35th International Colloquium on Automata, Languages,
  and Programming (ICALP)}, pages 253--266, 2006.

\bibitem{APTT03}
Aaron Archer, Christos Papadimitriou, Kunal Talwar, and \'{E}va Tardos.
\newblock An approximate truthful mechanism for combinatorial auctions with
  single parameter agents.
\newblock In {\em SODA '03: Proceedings of the fourteenth annual ACM-SIAM
  symposium on Discrete algorithms}, pages 205--214, 2003.

\bibitem{Armstrong}
Mark Armstrong.
\newblock Price discrimination by a many-product firm.
\newblock {\em Review of Economic Studies}, 66(1):151--68, January 1999.

\bibitem{BB06}
Maria-Florina Balcan and Avrim Blum.
\newblock Mechanism design, machine learning, and pricing problems.
\newblock {\em SIGecom Exch.}, 7(1):34--36, 2007.

\bibitem{Briest}
Patrick Briest.
\newblock Uniform budgets and the envy-free pricing problem.
\newblock In {\em Proc. 35th International Colloquium on Automata, Languages,
  and Programming (ICALP)}, pages 808--819, 2008.

\bibitem{CHK07}
Shuchi Chawla, Jason~D. Hartline, and Robert Kleinberg.
\newblock Algorithmic pricing via virtual valuations.
\newblock In {\em Proc. 8th ACM Conference on Electronic Commerce (EC)}, pages
  243--251, 2007.

\bibitem{D3R08}
Peerapong Dhangwatnotai, Shahar Dobzinski, Shaddin Dughmi, and Tim Roughgarden.
\newblock Truthful approximation schemes for single-parameter agents.
\newblock In {\em FOCS '08: Proceedings of the 2008 49th Annual IEEE Symposium
  on Foundations of Computer Science}, pages 15--24, 2008.

\bibitem{DNS06}
Shahar Dobzinski, Noam Nisan, and Michael Schapira.
\newblock Truthful randomized mechanisms for combinatorial auctions.
\newblock In {\em STOC '06: Proceedings of the thirty-eighth annual ACM
  symposium on Theory of computing}, pages 644--652, 2006.

\bibitem{envyfree}
Venkatesan Guruswami, Jason~D. Hartline, Anna~R. Karlin, David Kempe, Claire
  Kenyon, and Frank McSherry.
\newblock On profit-maximizing envy-free pricing.
\newblock In {\em Proc. 16th ACM-SIAM Symposium on Discrete Algorithms (SODA)},
  pages 1164--1173, 2005.

\bibitem{AGT-HK}
Jason Hartline and Anna Karlin.
\newblock Profit maximization in mechanism design.
\newblock In Eva~Tardos Noam~Nisan, Tim~Roughgarden and Vijay~V. Vazirani,
  editors, {\em Algorithmic Game Theory}, chapter~13, pages 331--362. Cambridge
  Press, 2007.

\bibitem{MV}
Alejandro~M. Manelli and Daniel~R. Vincent.
\newblock Bundling as an optimal selling mechanism for a multiple-good
  monopolist.
\newblock {\em Journal of Economic Theory}, 127(1):1 -- 35, 2006.

\bibitem{MV-2}
Alejandro~M. Manelli and Daniel~R. Vincent.
\newblock Multidimensional mechanism design: Revenue maximization and the
  multiple-good monopoly.
\newblock {\em Journal of Economic Theory}, 137(1):153--185, November 2007.

\bibitem{McAfeeMcMillan}
R.~Preston McAfee and John McMillan.
\newblock Multidimensional incentive compatibility and mechanism design.
\newblock {\em Journal of Economic Theory}, 46(2):335--354, December 1988.

\bibitem{MitzenmacherUpfal}
Michael Mitzenmacher and Eli Upfal.
\newblock {\em Probability and Computing: Randomized Algorithms and
  Probabilistic Analysis.}
\newblock Cambridge University Press, 2005.

\bibitem{Mye}
Roger Myerson.
\newblock Optimal auction design.
\newblock {\em Mathematics of Operations Research}, 6:58--73, 1981.

\bibitem{RZ}
John Riley and Richard Zeckhauser.
\newblock Optimal selling strategies: {W}hen to haggle, when to hold firm.
\newblock {\em Quarterly J. Economics}, 98(2):267--289, 1983.

\bibitem{RochetChone}
Jean-Charles Rochet and Philippe Chone.
\newblock Ironing, sweeping, and multidimensional screening.
\newblock {\em Econometrica}, 66(4):783--826, July 1998.

\bibitem{Than}
John Thanassoulis.
\newblock Haggling over substitutes.
\newblock {\em J. Economic Theory}, 117:217--245, 2004.

\end{thebibliography}

\newpage

\begin{appendix}

\section{The Buy-One Model in Dimension $2$}
\label{a:dim2}

In the $2$-dimensional setting, we will refer to the two distinct items as the $x$-item and $y$-item. A given consumer valuation is a point $(v_x,v_y)\in (\R _0^+)^2$. For a given lottery with probabilities $(\phi _x,\phi _y)$ and $\phi _x+\phi _y=1$ (an assumption we make throughout this section) and price $p$ we call the set of all $(x,y)$ with $\phi_x x + \phi_y y = p$ its {\em indifference line}. This corresponds to the set of consumer valuations that result in zero utility from buying this lottery. Geometrically, it is a line with $x$-intercept $p/\phi_x$ and $y$-intercept $p/\phi_y$. Similarly, any arbitrary line in the plane with positive $x$- and $y$-intercepts corresponds to a lottery. The following Lemma gives an easy geometric interpretation of a lottery's price.

\begin{lemma}
\label{l:geomPrice}
Given a line with positive $x$- and $y$-intercepts, the price of the corresponding lottery is the $x$-coordinate (or $y$-coordinate) of its intersection with the diagonal $x=y$.
\end{lemma}
\begin{proof}
Consider lottery $\lambda =((\phi _x,\phi _y),p)$ and the valuation $(v_x,v_y)$ at the intersection of its indifference line with the diagonal. Since $v_x=v_y$ and $\phi _x+\phi _y=1$, this consumer's value for the lottery is $v_x$.  In addition, because she is on the indifference line of $\lambda$, we have $v_x=p$.
\end{proof}

Let a consumer distribution $\cd$ and a system of lotteries $\Lambda$ be given. Consider a consumer with valuation $(v_x,v_y)$ from $\cd$ who receives utility $u$ from the lottery she picks from $\Lambda$ and let $\delta =\min \{ u,v_x,v_y\}$. We can replace this consumer's valuation with $(v_x-\delta ,v_y-\delta )$ without affecting her selection, since her expected utility from every lottery has decreased by exactly $\delta$ or all the way down to $0$. In particular, if we were to modify the lottery system, it would be guaranteed that as long as the consumer with valuations $(v_x-\delta ,v_y-\delta )$ can afford to buy any lottery, the same lottery would be bought by a consumer with valuation $(v_x,v_y)$.

We can use this observation to significantly simplify the kinds of consumer distributions we have to consider. Given $\cd$ and $\Lambda$, we define $\cd '$ as above by replacing each consumer type $(v_x,v_y)$ in $\cd$ with $(v_x-\delta ,v_y-\delta )$. As argued before, this does not change the revenue of lottery system $\Lambda$ and for every lottery system (or item pricing) it holds that its revenue on $\cd$ is an upper bound on its revenue on $\cd '$.

By construction every consumer in $\cd '$ has either utility $0$ from the lottery she purchases or has valuation $0$ for either the $x$- or $y$-item. Geometrically, all consumer types lie either on the positive part of one of the coordinate axes or on the boundary of the set of consumer types that cannot afford to purchase any lottery from $\Lambda$. Geometrically, the set of consumer types unable to afford any lottery is the intersection of the halfspaces defined by the lotteries in $\Lambda$ and, consequently, its boundary is the intersection of a convex polygon with the coordinate system's first quadrant. This situation is depicted in Fig. \ref{fig:2dim}.

We denote by $\cd '_x$, $\cd '_y$, and $\cd '_0$ the consumer types from distribution $\cd '$ that lie on the $x$-axis, $y$-axis, or on the {\em indifference polygon} defined by the indifference lines, respectively. Let $x^*$ and $y^*$ denote the $x$- and $y$-intersects of the indifference polygon. The following algorithm turns lottery system $\Lambda$ into a pure item pricing.

\begin{itemize}
\item[(1)] With probability $1/3$, split every lottery $\lambda =((\phi _x,\phi _y),p)\in \Lambda$ into two lotteries $\lambda _x=((\phi _x,0),p)$, $\lambda _y=((0,\phi _y),p)$ and let $\Lambda _x=\{ \lambda _x\, |\, \lambda \in \Lambda\}$, $\Lambda _y=\{ \lambda _y\, |\, \lambda \in \Lambda\}$. Independently run the randomized rounding procedure described in Theorem \ref{t:dim1} on both the lotteries in $\Lambda _x$ and $\Lambda _y$.
\item[(2)] With probability $2/3$ choose $\delta$, such that $(x^*/2+\delta ,y^*/2+\delta )$ lies on the indifference polygon. (Note, that $(x^*/2,y^*/2)$ lies inside the polygon by convexity and, thus, $\delta$ is well defined.) Assign prices $p_x=x^*/2+\delta$ and $p_y=y^*/2+\delta$ to the items.
\end{itemize}

We proceed by analyzing the expected revenue of the returned item pricing. For consumers in $\cd '_x$ and $\cd '_y$ it immediately follows from the analysis in Theorem \ref{t:dim1} that their expected payment given the pure item prices is exactly the same as given lottery system $\Lambda$. Since Step (1) of the algorithm is performed with probability $1/3$, the expected revenue from consumers in $\cd '_x$ and $\cd '_y$ is decreased by a factor of $3$.

Let us then fix a consumer from $\cd '_0$ with valuations $(v_x,v_y)$ located somewhere on the indifference polygon. By Lemma \ref{l:geomPrice} the price paid by this consumer given the lottery system is the $x$-coordinate of the intersection of the tangent line to the indifference polygon in $(v_x,v_y)$ with the diagonal $x=y$. We consider three cases.

Assume first that $v_x<p_x$. Then, given the item pricing, the consumer will choose to buy the $y$-item at price $p_y=y^*/2+\delta \ge y^*/2$. Recall that we assume w.l.o.g. that $y^*\ge x^*$. Since by convexity no tangent line can intersect the diagonal above $y^*$, the consumer's payment decreases by at most a factor of $2$.

Let then $v_x\ge p_x$. Assume first that $v_x\ge v_y$. By convexity, no tangent line to the polygon in this region can intersect the diagonal above $x^*$ and the same argument as in the previous case bounds our loss in revenue by a factor of $2$.

Finally, assume that $v_x\ge p_x$ and $v_x<v_y$. The tangent line to the polygon in point $(p_x,p_y)$ can be written as $y=y^*-((y^*-p_y)/p_x)x$. By convexity, $(v_x,v_y)$ lies below this line and the tangent in this point has smaller slope. Consequently, its intersection with the diagonal is upper bounded by the intersection of the tangent line in $(p_x,p_y)$ with the diagonal, which is defined by the equation\[
x=y^*-\frac{y^*-p_y}{p_x}x.\]
Observe that $p_y-p_x=y^*/2-x^*/2\le y^*/2$. We obtain that
\begin{eqnarray*}
x & = & \left( 1+\frac{y^*-p_y}{p_x}\right) ^{-1}y^*\\
 & = & \frac{y^*}{y^*-(p_y-p_x)}p_x \le 2p_x,
\end{eqnarray*}
which again yields a bound of $2$ on the loss in revenue. Since Step (2) of the algorithm is performed with probability $2/3$, we expect a total loss of at most a factor $3$.

\vspace{7mm}

\noindent {\bf Theorem \ref{t:dim2}} {\em
For $n = 2$ it holds that $r^*_L(\cd ) \leq 3r^*(\cd )$ for any consumer distribution $\cd$.
}

\section{Uniform Valuations in the Buy-One Model}
\label{a:uniVals}

Let $\cd$ be a distribution over uniform-valuation consumers. Formally, every consumer type in $\cd$ is of the form $(S,v)$, where $S$ is a subset of the items and $v$ the value for any item from $S$. The value for items from the complement of $S$ is $0$.

{\bf Upper Bound.} Let a consumer distribution $\cd$ and a lottery system $\Lambda$ be given. We partition $\cd$ into partial distributions $\cd _S$ for all possible item sets $S$, where $\cd _S$ contains all consumer types interested in item set $S$. Pick one $\cd _T$ uniformly at random and consider the set of one-dimensional lotteries $\Lambda _{T}=\{ (\sum _{i\in T}\phi _i,p)\, |\, ((\phi _1,\ldots ,\phi _n),p)\in \Lambda \}$. Apply the randomized rounding procedure described in Theorem \ref{t:dim1} to $\Lambda _T$ and assign the resulting price to all items.

As shown in Theorem \ref{t:dim1} the revenue from consumer distribution $\cd _T$ does not decrease and, since every $\cd _S$ is picked with probability $1/(2^{n}-1)$, the overall expected decrease in revenue is bounded above by a factor of $2^n$.

{\bf Lower Bound.} Let $k=2^n-2$ and $S_0,\ldots ,S_k$ be all distinct non-empty subsets of $[n]$ ordered by decreasing cardinality, i.e., $|S_0|\ge |S_1|\ge \cdots \ge |S_k|$. We define a consumer distribution as follows. For all $0\le j\le k$ we have a consumer type $c_j=(S_j,n^j)$ with probability $n^{-j}((1-n^{-1})/(1-n^{-k-1}))$. Similarly, lottery system $\Lambda$ contains a lottery $\lambda _j$ for each $0\le j\le k$, where $\lambda _j$ has probability $1/|S_j|$ for each item in $S_j$ and price $n^{j-1}$. The utility $u(c_j,\lambda _j)$ of consumer type $c_j$ from buying lottery $\lambda _j$ is $u(c_j,\lambda _j)=n^j-n^{j-1}$. For $i<j$ we know that $|S_i\backslash S_j|\ge 1$ and, thus, the probability that lottery $\lambda _i$ allocates an item from set $S_j$ is at most $1-1/n$. Thus, $c_j$'s utility from buying $\lambda _i$ is bounded above by $u(c_j,\lambda _i)\le (1-1/n)n^j=n^j-n^{j-1}$. Finally, lotteries $\lambda _i$ with $i>j$ are too expensive for this consumer type to afford. Thus, we may assume that each consumer type $c_j$ chooses to purchase lottery $\lambda _j$ when offered $\Lambda$. The total revenue obtained by $\Lambda$ then can be written as\[
\sum _{j=0}^kn^{-j}\frac{1-n^{-1}}{1-n^{-k-1}}n^{j-1}\ge \frac{1}{n}(2^n-2).\]
On the other hand, consider any pure item pricing. If the price of some item falls into the interval $(n^{j-1},n^j]$, then the total probability mass of consumer types able to afford this item is bounded above by $2n^{-j}$ and, thus, the total revenue from this item cannot exceed $n^j\cdot 2n^{-j}=2$. Summing over all items yields a bound of $2n$ on the overall revenue obtainable by any item pricing.

\vspace{7mm}

\noindent {\bf Theorem \ref{t:uniVals}} {\em
Let $\cd$ be a distribution on uniform valuation consumers. Then $r^*_L(\cd )/r^*(\cd )=\bigo (2^n)$. There exist distributions $cd$ with $r^*_L(\cd )/r^*(\cd )=\tilde{\Omega }(2^n)$.
}

\section{Vector Packing}
\label{a:vectorPacking}

{\bf Lemma \ref{t:vectorPacking}} {\em
Let $n\ge 1$ be given. For every $q\ge 2n$ there exists a set $\mathcal{V}^n_q$ of vectors in $S^{n+}_{1/\sqrt{n}}$, such that $v\cdot w\le 1/n-1/q$ for all $v,w\in \mathcal{V}^n_q$ with $v\not= w$ and $|\mathcal{V}^n_q|=\Omega (q^{(n-1)/2})$.
}

\begin{proof}
For vectors $v,w\in S^{n+}_{1/\sqrt{n}}$ we may write that\[
v\cdot w=\frac{1}{2}\left( v^2+w^2-(v-w)^2\right) = \frac{1}{n}-\frac{1}{2}||v-w||^2_2.\]
Consequently, the condition that $v\cdot w\le 1/n-1/q$ for all $v,w\in \mathcal{V}^n_q$ is equivalent to asking that a ball of radius $\sqrt{2/q}$ around the tip of any vector from $\mathcal{V}^n_q$ does not contain the tip of any of the other vectors. For the remainder of this proof we will associate vectors in $\mathcal{V}^n_q$ with $n$-dimensional balls of radius $\sqrt{2/q}$ centered at their tips.

We construct the set $\mathcal{V}^n_q$ by the following simple greedy approach. While there exists a point in $S^{n+}_{1/\sqrt{n}}$ that is not covered by previously selected balls, we choose it as the center of a new ball of radius $\sqrt{2/q}$.

We now have to lower bound the number of vectors found in this fashion. When the procedure terminates, it must be the case that $S^{n+}_{1/\sqrt{n}}$ is completely covered by the selected balls. Choose any $\varepsilon >0$. By $B^n_r$ we denote the $n$-dimensional ball of radius $r$ centered at the origin, $B^{n+}_r$ its part in the all-positive orthant. If we increase the radius of all balls chosen by our packing procedure by $\varepsilon \sqrt{2/q}$, we know that they completely cover the set\[
B^{n+}_{1/\sqrt{n}+\varepsilon \sqrt{2/q}}-B^{n+}_{1/\sqrt{n}-\varepsilon \sqrt{2/q}},\]
as all points in this set are at distance at most $\varepsilon \sqrt{2/q}$ from $S^{n+}_{1/\sqrt{n}}$. The $n$-dimensional Lebesgue measure of this set is\[
\lambda ^n(B^{n+}_{1/\sqrt{n}+\varepsilon \sqrt{2/q}}-B^{n+}_{1/\sqrt{n}-\varepsilon \sqrt{2/q}})=2^{-n}\left( (\sqrt{1/n}+\varepsilon \sqrt{2/q})^n - (\sqrt{1/n}-\varepsilon \sqrt{2/q})^n\right) \lambda ^n(B^n_1),\]
where $\lambda ^n(B^n_1)$ denotes the measure of the $n$-dimensional unit ball. Similarly, each ball of radius $(1+\varepsilon )\sqrt{2/q}$ has measure\[
\left((1+\varepsilon )\sqrt{2/q}\right) ^n \lambda ^n(B^n_1),\]
and it follows that the number of balls selected by our procedure is at least
\begin{eqnarray*}
 & & \frac{2^{-n}\left( \left( \sqrt{1/n}+\varepsilon \sqrt{2/q}\right) ^n - \left( \sqrt{1/n}-\varepsilon \sqrt{2/q}\right) ^n\right)}{\left( (1+\varepsilon )\sqrt{2q}\right) ^n}\\
 & \ge & \frac{2^{-n+1}\varepsilon \sqrt{2/q}\cdot n\left( \sqrt{1/n}-\varepsilon \sqrt{2/q}\right) ^{n-1}}{\left( (1+\varepsilon )\sqrt{2q}\right) ^n}\\
 & \ge & \frac{2^{-n+1}\varepsilon \sqrt{2/q}\cdot n\left( (1-\varepsilon )\sqrt{1/n}\right) ^{n-1}}{\left( (1+\varepsilon )\sqrt{2q}\right) ^n}\\
 & = & \frac{\varepsilon n\left( (1-\varepsilon )\sqrt{1/n}\right) ^{n-1}}{(1+\varepsilon )^n2^{(3n-3)/2}}q^{(n-1)/2} = \Omega (q^{(n-1)/2}),
\end{eqnarray*}
where the first inequality follows by lower-bounding the difference between the two terms of the enumerator by $2\varepsilon \sqrt(2/q)$ times the derivative of the convex function $x^n$ in $\sqrt{1/n}-\varepsilon \sqrt{2/q}$ and the second inequality uses the fact that $\sqrt{2/q}\le \sqrt{1/n}$.
\end{proof}

\end{appendix}

\end{document}